\documentclass[11pt]{article}
\usepackage{amssymb}

\usepackage{graphicx}
\usepackage{amsmath}
\usepackage{makeidx}
\usepackage{indentfirst}

\newcounter{resultnum}[section]

\newcounter{conclusionnum}[section]

\newcounter{conditionnum}[section]

\newcounter{conjecturenum}[section]

\newcounter{examplenum}[section]

\newcounter{exercisenum}[section]
\newtheorem{lemma}{Lemma}[section]

\newcounter{lemmanum}[section]

\newcounter{notationnum}[section]
\newtheorem{theorem}{Theorem}[section]

\newcounter{theoremnum}[section]
\newtheorem{definition}{Definition}[section]

\newcounter{definitionnum}[section]
\newtheorem{corollary}{Corollary}[section]

\newcounter{corollarynum}[section]
\newtheorem{remark}{Remark}[section]

\newcounter{remarknum}[section]
\newtheorem{proposition}{Proposition}[section]

\newcounter{propositionnum}[section]

\newcounter{acknowledgementnum}[section]

\newcounter{algorithmnum}[section]

\newcounter{axiomnum}[section]

\newcounter{casenum}[section]

\newcounter{claimnum}[section]

\newcounter{summarynum}[section]

\newcounter{problemnum}[section]
\newenvironment{proof}[1][]{\textbf{Proof.} }{}

\begin{document}

\title{On General Solutions of Einstein Equations}

\date{\small Submitted to arXiv.org on September 22, 2009;\newline
Published journal version in IJGMMP 8 (2011) 9-21, on March 17, 2011,\newline equivalent to arXiv: 0909.3949v1  [gr-qc];  \newline
an extended/modified variant published in IJTP 49 (2010) 884-913,\newline equivalent to arXiv: 0909.3949v4 [gr-qc]; \newline
on June 20, 2011, moderators arXiv.org accepted to provide a different number to this "short" variant in physics.gen-ph)}
\author{Sergiu I. Vacaru\thanks{
sergiu.vacaru@uaic.ro, Sergiu.Vacaru@gmail.com;\newline
http://www.scribd.com/people/view/1455460-sergiu } \\
{\quad} \\
{\small {\textsl{\ Science Department, University "Al. I. Cuza" Ia\c si},} }%
\\
{\small {\textsl{\ 54 Lascar Catargi street, 700107, Ia\c si, Romania}} }}
\maketitle

\begin{abstract}
We show how the Einstein equations with cosmological constant (and/or various types of matter field sources) can be integrated in a very general form following the anholonomic deformation method for constructing exact solutions
in four and five dimensional gravity (S. Vacaru, IJGMMP {\bf 4} (2007) 1285). In this letter, we prove that such a geometric method  can be used for constructing general non--Killing solutions. The key idea is to introduce an auxiliary linear connection which is also metric compatible and completely defined by the metric structure but contains some torsion terms induced nonholonomically by  generic off--diagonal coefficients  of metric. There are some classes of nonholonomic frames with respect to which the Einstein equations (for such an auxiliary  connection) split into an integrable  system  of partial differential equations. We have to impose  additional constraints on generating and integration functions in order to transform the auxiliary connection into the Levi--Civita one. This way, we  extract general exact solutions (parametrized by generic off--diagonal metrics and depending on all coordinates) in Einstein gravity and five dimensional extensions.

\vskip5pt



\end{abstract}

To construct the most general classes of metrics solving the gravitational field equations in Einstein gravity and extra dimension generalizations is of considerable importance in modern gravity, cosmology and astrophysics. This is a very difficult mathematical task because of high complexity of
such systems of nonlinear partial differential equations. Various types of numerical and analytic approaches have not attempted to solve the problem in a general form but oriented to some particular types of exact or approximate solutions which seem to be of physical interest (black holes, cosmological
solutions, nonlinear gravitational waves etc). Surprisingly, there were
elaborated certain geometric methods which allows us to represent the field
equations for various types of gravitational field theories in some convenient (for further integration) forms. Following this approach,  to generate exact solutions with generic off--diagonal metrics, nonholonomic frames\footnote{the word nonholonomic, equivalently, anholonomic means that our geometric constructions will be adapted with respect to certain classes of nonholonomic/nonintegrable frames} and various types o linear connections  became a question of frame
transforms and constraining integral varieties for corresponding systems of
partial differential equations which can be solved in very general forms.

In the present paper, we prove that the Einstein equations with certain type
of general sources (in particular, with nonzero, or vanishing, cosmological
constants) can be solved following the anholonomic deformation method, see original results and reviews in Refs. \cite{ijgmmp,vncg,vsgg}.\footnote{we use anholonomic deformations of frame, metric and connection structures which makes our approach more general than the Cartan's moving frame method when the same fundamental geometric objects are equivalently re--defined with respect to certain convenient systems of reference; our idea is to solve the problem for a more general connection, also defined by the same metric structure in a unique metric compatible form, and than to constrain  the solutions to generate Levi--Civita configurations } In our
approach, we use the nonlinear connection formalism originally developed in
Finsler and Lagrange geometry but recently modified for applications in
general relativity and some 'standard' models of quantum gravity,
noncommutative Ricci flow theory and string/brane gravity models on
nonholonomic (pseudo) Riemanian and Riemann--Cartan manifolds. Such
constructions were elaborated, for instance, in Refs. \cite%
{vrflg,vspdo,vbrane} following geometric ideas originally considered for
vector and tangent bundles \cite{ma1987,ma}.

We emphasize that in this work the
metrics and connections do not depend on ''velocities'', i.e. we do not work
with geometric objects on tangent bundles, even a number of analogies with
constructions in Lagrange--Finsler geometry can be found. All results can be
stated for four dimensional, 4--d, (pseudo) Riemannian manifolds. Extensions
to 5--d Einstein manifolds, with conventional $3+2$ splitting of dimensions,
and nonholonomic reductions to 2+2, will be used only because they simplify
proofs of results and show explicitly how the anholonomic deformation method of
constructing exact solutions can be generalized.

Let us consider a (pseudo) Riemannian 5--d manifold $\ ^{5}\mathbf{V}$
endowed with a metric $\mathbf{g}=g_{\alpha \beta }(u^{\gamma })du^{\alpha
}\otimes du^{\beta }$ of arbitrary signature $\epsilon _{\alpha }=(\epsilon
_{1}=\pm 1,\ldots ,\epsilon _{5}=\pm 1).$\footnote{%
In our works, we follow conventions from \cite{ijgmmp,vrflg} when left up/low
indices are used as labels for geometric spaces and objects.} The local
coordinates on $^{5}\mathbf{V}$ are parametrized in the form $u^{\alpha
}=(x^{i},y^{a}),$ where $x^{i}=(x^{1},x^{\widehat{i}})$ and $y^{a}=\left(
v,y\right) ,$ i. e. $y^{4}=v,$ $y^{5}=y.$ Indices $i,j,k,...=1,2,3;$ $\hat{%
\imath},\hat{\jmath},\hat{k}...=2,3$ and $a,b,c,...=4,5$ are used for a
conventional $(3+2)$--splitting of dimension and general abstract/coordinate
indices when $\alpha ,\beta ,\ldots $ run values $1,2,...,5$ $.$ For 4--d
constructions, we can write $\ ^{4}\mathbf{V}$ and $u^{\widehat{\alpha }%
}=(x^{\widehat{i}},y^{a}),$ when the coordinate $x^{1}$ and values for
indices like $\alpha ,i,...=1$ are not considered. In brief, we shall denote
some partial derivatives $\partial _{\alpha }=\partial /\partial u^{\alpha }$
in the form $s^{\bullet }=\partial s/\partial x^{2},s^{\prime }=\partial
s/\partial x^{3},s^{\ast }=\partial s/\partial y^{4}.$

We write $\nabla =\{\Gamma _{\ \beta \gamma }^{\alpha }\}$ for the
Levi--Civita connection,\footnote{%
which is uniquely defined by a given tensor $\mathbf{g}$ to be metric
compatible, $\nabla \mathbf{g}=0,$ and with zero torsion;  we
  summarize on ''up-low'' repeating indices if the contrary is
not stated} with coefficients stated with respect to an arbitrary local
frame basis $e_{\alpha }=(e_{i},e_{a})$ and its dual basis $e^{\beta
}=(e^{j},e^{b}).$ Using the Riemannian curvature tensor $\mathcal{R}=\{R_{\
\beta \gamma \delta }^{\alpha }\}$ defined by $\nabla ,$ one constructs the
Ricci tensor, $\mathcal{R}ic=\{R_{\ \beta \delta }\doteqdot R_{\ \beta
\alpha \delta }^{\alpha }\},$ and scalar curvature $R\doteqdot g^{\beta
\delta }R_{\ \beta \delta },$ where $g^{\beta \delta }$ is inverse to $%
g_{\alpha \beta }.$ The Einstein equations on $\mathbf{V,}$ for an
energy--momentum source $T_{\alpha \beta },$ are written in the form%
\begin{equation}
R_{\ \beta \delta }-\frac{1}{2}g_{\beta \delta }R=\varkappa T_{\beta \delta
},  \label{einsteq}
\end{equation}%
where $\varkappa =const.$ For the Einstein spaces defined by a cosmological
constant $\lambda ,$ such gravitational field equations can be represented
as $R_{\ \ \beta }^{\alpha }=\lambda \delta _{\beta }^{\alpha },$ where $%
\delta _{\beta }^{\alpha }$ is the Kronecher symbol. The vacuum solutions
are obtained for $\lambda =0.$

The goal of this paper is to formulate and sketch the proof of (Main Result):

\begin{theorem}
\label{mth}If the gravitational field equations in Einstein gravity and its
5-d extension (\ref{einsteq}) can be represented via frame transforms in the
form
\begin{equation}
R_{\ \ \beta }^{\alpha }=\Upsilon _{\ \ \beta }^{\alpha }  \label{einst1}
\end{equation}
for a given $\Upsilon _{\ \ \beta }^{\alpha }=diag[\Upsilon _{\gamma }]$
with
\begin{equation}
\Upsilon _{1}=\Upsilon _{2}+\Upsilon _{4},\ \Upsilon _{\ \ 2}^{2}=\Upsilon
_{\ \ 3}^{3}=\Upsilon _{2}(x^{k},v),\Upsilon _{\ \ 4}^{4}=\Upsilon _{\ \
5}^{5}=\Upsilon _{4}(x^{\widehat{k}}),  \label{source}
\end{equation}%
for $\ y^{4}=v,$ such equations can be solved in general form by metrics of
type
\begin{eqnarray}
\ ^{5}\mathbf{g} &\mathbf{=}&\epsilon _{1}{dx^{1}\otimes dx^{1}}+g_{\widehat{%
i}}(x^{\widehat{k}}){dx^{\widehat{i}}\otimes dx^{\widehat{i}}}+\omega
^{2}(x^{j},y^{b})h_{a}(x^{k},v)\mathbf{e}^{a}{\otimes }\mathbf{e}^{a},
\nonumber \\
\mathbf{e}^{4} &=&dy^{4}+w_{i}(x^{k},v)dx^{i},\mathbf{e}%
^{5}=dy^{5}+n_{i}(x^{k},v)dx^{i},  \label{ansgensol}
\end{eqnarray}%
where coefficients are determined by generating functions $%
f(x^{i},v),f^{\ast }\neq 0,$ and $\omega (x^{j},y^{b})\neq 0$ and
integration functions $\ ^{0}f(x^{i}),\ ^{0}h(x^{i}),$ $\ ^{1}n_{k}(x^{i})$
and $\ ^{2}n_{k}(x^{i}),$ following formulas%
\begin{eqnarray}
g_{\widehat{i}} &=&\epsilon _{\widehat{i}}e^{\psi (x^{\widehat{k}})},%
\mbox{\
for }\epsilon _{2}\psi ^{\bullet \bullet }+\epsilon _{3}\psi ^{\prime \prime
}=\Upsilon _{4};  \label{coeff} \\
h_{4} &=&\epsilon _{4}\ ^{0}h(x^{i})\ [f^{\ast }(x^{i},v)]^{2}|\varsigma
(x^{i},v)|\mbox{\ and }h_{5}=\epsilon _{5}[f(x^{i},v)-\ ^{0}f(x^{i})]^{2};
\nonumber \\
w_{i} &=&-\partial _{i}\varsigma (x^{i},v)/\varsigma ^{\ast }(x^{i},v)%
\mbox{\ and }  \nonumber \\
n_{k} &=&\ ^{1}n_{k}(x^{i})+\ ^{2}n_{k}(x^{i})\int dv\ \varsigma
(x^{i},v)[f^{\ast }(x^{i},v)]^{2}/[f(x^{i},v)-\ ^{0}f(x^{i})]^{3},  \nonumber \\
\mbox{\ for }\varsigma &=&\ ^{0}\varsigma (x^{i})-\frac{\epsilon _{4}}{8}\
^{0}h(x^{i})\int dv\ \Upsilon _{2}(x^{k},v)\ f^{\ast }(x^{i},v)[f(x^{i},v)-\
^{0}f(x^{i})];  \nonumber \\
\mathbf{e}_{k}\omega &=&\partial _{k}\omega +w_{k}\omega ^{\ast
}+n_{k}\partial \omega /\partial y^{5}=0,  \label{confcond}
\end{eqnarray}%
when the so--called Levi--Civita integral varieties are selected by
additional constraints%
\begin{equation}
w_{i}^{\ast }=\mathbf{e}_{i}\ln |h_{4}|,\mathbf{e}_{k}w_{i}=\mathbf{e}%
_{i}w_{k},\ n_{i}^{\ast }=0,\ \partial _{i}n_{k}=\partial _{k}n_{i}.
\label{lccond}
\end{equation}
\end{theorem}

In order to construct some explicit classes of exact solutions of Einstein
equations (\ref{einst1}), we have to state certain boundary/ symmetry/
topology conditions which would allow to define the integration functions
and systems of first order partial differential equations of type (\ref%
{lccond}). Perhaps all classes of exact solutions outlined in Refs. \cite%
{kramer,bic,ijgmmp,vncg,vsgg} can be found as certain particular cases of
metrics (\ref{ansgensol}) or equivalently redefined in such a form.

\begin{remark}
\begin{enumerate}
\item Analogs of Theorem \ref{mth} were proven in our previous works \cite%
{ijgmmp,vncg,vsgg} for certain cases with $\omega =1,$ and other types
generalizations, which allowed us to generate various classes of generic
off--diagonal\footnote{%
which can not be diagonalized by coordinate transform} exact solutions with
one Killing vector (in such a case, metrics (\ref{ansgensol}) do not depend
on variable $y^{5}).$ The key new result of this work is that we can
consider any generating function $\omega (x^{j},y^{b})$ depending on
coordinate $y^{5}$ but subjected to the condition (\ref{confcond}). This
allows us to construct very general classes of ''non--Killing'' exact
solutions.

\item It should be emphasized that any (pseudo) Riemannian metric $\mathbf{g}%
=\{g_{\alpha ^{\prime }\beta ^{\prime }}(u^{\alpha ^{\prime }})\}$ depending
in general on all five local coordinates on $\ ^{5}\mathbf{V}$ can be
parametrized in a form $g_{\alpha \beta }$ (\ref{ansgensol}), $g_{\alpha
\beta }=e_{\ \alpha }^{\alpha ^{\prime }}e_{\ \beta }^{\beta ^{\prime
}}g_{\alpha ^{\prime }\beta ^{\prime }},$ using frame transforms of type $%
e_{\alpha }=e_{\ \alpha }^{\alpha ^{\prime }}e_{\alpha ^{\prime }}.$%
\footnote{%
We have to solve certain systems of quadratic algebraic equations and define
some $e_{\ \alpha }^{\alpha ^{\prime }}(u^{\beta }),$ choosing a convenient
system of coordinates $u^{\alpha ^{\prime }}=u^{\alpha ^{\prime }}(u^{\beta
}).$} So, the metrics constructed above define general solutions of Einstein
equations for any type of sources $\varkappa T_{\beta \delta } $ which can
be parametrized in a formally diagonalized form (\ref{source}), with respect
to a nonholonomic frame of reference\footnote{%
using chains of frame transforms, such parametrizations can be defined for
'almost' all physically important energy--momentum tensors}.
\end{enumerate}
\end{remark}

Let us provide the key points for a proof of Theorem \ref{mth} following
Steps 1--6 of the anholonomic deformation/ frame method (proposed in Refs. \cite%
{vpoland,vapny,vhep2}, see recent reviews and generalizations in Refs. \cite%
{ijgmmp,vncg,vsgg,vrflg}):

\subsubsection*{Step 1: Ansatz for metrics and N--adapted frames}

We can consider a nonholonomic $(3+2)$--splitting of a spacetime $\ ^{5}%
\mathbf{V}$ by introducing a non--integrable distribution stated by certain
coefficients $\mathbf{N}=\{N_{i}^{a}\},$ when $\mathbf{N}=N_{i}^{a}(u^{%
\alpha })dx^{i}\otimes \frac{\partial }{\partial y^{a}}.$ This defines a
class of so--called N--adapted frames, (respectively) dual frames
\begin{eqnarray}
\mathbf{e}_{\nu } &=&\left( \mathbf{e}_{i}=\frac{\partial }{\partial x^{i}}%
-N_{i}^{a}\frac{\partial }{\partial y^{a}},e_{a}=\partial _{a}=\frac{%
\partial }{\partial y^{a}}\right) ,  \label{dder} \\
\mathbf{e}^{\mu } &=&\left( e^{i}=dx^{i},\mathbf{e}%
^{a}=dy^{a}+N_{i}^{a}dx^{i}\right) .  \label{ddif}
\end{eqnarray}%
The vielbeins (\ref{ddif}) satisfy the nonholonomy relations
\begin{equation}
\lbrack \mathbf{e}_{\alpha },\mathbf{e}_{\beta }]=\mathbf{e}_{\alpha }%
\mathbf{e}_{\beta }-\mathbf{e}_{\beta }\mathbf{e}_{\alpha }=w_{\alpha \beta
}^{\gamma }\mathbf{e}_{\gamma }  \label{anhrel}
\end{equation}%
with (antisymmetric) nontrivial anholonomy coefficients $w_{ia}^{b}=\partial
_{a}N_{i}^{b}$ and $w_{ji}^{a}=\Omega _{ij}^{a},$ where\footnote{%
in Lagrange--Finsler geometry, for $\ ^{4}\mathbf{V}$ $=TM,$ where $TM$ is
the total space of a tangent bundle on a manifold $M,$ such a $\mathbf{N}$
defines a nonlinear connection (N--connection) structure \cite{ma1987,ma};
nevertheless, N--connections can be considered on nonholonomic manifolds,
i.e. manifolds enabled with nonholonomic distributions, even in general
relativity, see discussions in \cite{ijgmmp,vsgg,vrflg}}
\begin{equation}
\Omega _{ij}^{a}=\mathbf{e}_{j}\left( N_{i}^{a}\right) -\mathbf{e}_{i}\left(
N_{j}^{a}\right)   \label{ncurv}
\end{equation}%
are the coefficients of N--connection curvature. The particular holonomic/
integrable case is selected by the integrability conditions $w_{\alpha \beta
}^{\gamma }=0.$\footnote{%
We use boldface symbols for spaces (and geometric objects on such spaces)
enabled with a structure of N--coefficients.}

Any (pseudo) Riemannian metric $\mathbf{g}$ on $\ ^{5}\mathbf{V}$ can be
written in the form \
\begin{equation}
\mathbf{g}=g_{ij}(u^{\alpha })e^{i}\otimes e^{j}+h_{ab}(u^{\alpha })\mathbf{e%
}^{a}\otimes \mathbf{e}^{b},  \label{gendm}
\end{equation}%
for some N--adapted coefficients $\left[ g_{ij},h_{ab}\right] $ and $%
N_{i}^{a}.$ For instance, we get the metric (\ref{ansgensol}) with $\omega
=1,$ from (\ref{gendm}) if we choose (omitting, for simplicity, priming of
indices)%
\begin{equation}
g_{ij}=diag[\epsilon _{1},g_{\widehat{i}}(x^{\widehat{k}%
})],h_{ab}=diag[h_{a}(x^{i},v)],N_{k}^{4}=w_{k}(x^{i},v),N_{k}^{5}=n_{k}(x^{i},v).
\label{data1}
\end{equation}%
Such a metric has a Killing vector, $e_{5}=\partial /\partial y^{5},$
symmetry because its coefficients do not depend on $y^{5}.$ Introducing a
nontrivial $\omega ^{2}(u^{\alpha })$ depending also on $y^{5},$ as a
multiple before $h_{a},$ we get a $(3+2)$ N--adapted parametrizaton, up to
certain frame/coordinate transforms, for all metrics on $\ ^{5}\mathbf{V}.$

\subsubsection*{Step 2: Metric compatible deformations of the Levi--Civita
connection}

It is a cumbersome task to prove using the Levi--Civita connection $\nabla $
(a unique one in general relativity being metric compatible, with zero
torsion, and completely defined by the metric structure) that the Einstein
equations (\ref{einst1}) are solved by metrics of type (\ref{ansgensol}). \
Our ''main trick'' is not only to adapt our constructions to N--adapted
frames of type $\mathbf{e}_{\alpha }$ (\ref{dder}) \ and $\mathbf{e}^{\mu }$
(\ref{ddif}) but also to use as an auxiliary tool (we emphasize, in Einstein
gravity and generalizations) a new type of linear connection $\widehat{%
\mathbf{D}}=\{\mathbf{\hat{\Gamma}}_{\ \beta \gamma }^{\alpha }\},$ also
uniquely defined by the metric structure. It can be defined as a 1--form $%
\widehat{\mathbf{\Gamma }}_{\ \beta }^{\alpha }=\widehat{\mathbf{\Gamma }}%
_{\ \beta \gamma }^{\alpha }\mathbf{e}^{\gamma }$ with coefficients $%
\widehat{\mathbf{\Gamma }}_{\ \alpha \beta }^{\gamma }=\left( \widehat{L}%
_{jk}^{i},\widehat{L}_{bk}^{a},\widehat{C}_{jc}^{i},\widehat{C}%
_{bc}^{a}\right)$ adapted to a $(3+2)$--splitting. Such a linear connection
is also metric compatible, $\widehat{\mathbf{D}}\mathbf{g=0,}$ defined by
any data $\mathbf{g=}\{g_{ij},h_{ab},N_{i}^{a}\}$ and contains an induced
torsion (by the same metric coefficients)
\begin{equation}
\widehat{\mathcal{T}}^{\alpha }=\widehat{\mathbf{T}}_{\ \beta \gamma
}^{\alpha }\mathbf{e}^{\beta }\wedge \mathbf{e}^{\gamma }\doteqdot \widehat{%
\mathbf{D}}\mathbf{e}^{\alpha }=d\mathbf{e}^{\alpha }+\widehat{\mathbf{%
\Gamma }}_{\ \beta }^{\alpha }\wedge \mathbf{e}^{\beta },  \label{tors}
\end{equation}%
with coefficients
\begin{eqnarray}
\widehat{T}_{\ jk}^{i} &=&\widehat{L}_{\ jk}^{i}-\widehat{L}_{\ kj}^{i},\
\widehat{T}_{\ ja}^{i}=-\widehat{T}_{\ aj}^{i}=\widehat{C}_{\ ja}^{i},\ T_{\
ji}^{a}=\Omega _{\ ji}^{a},\   \nonumber \\
\widehat{T}_{\ bi}^{a} &=&-\widehat{T}_{\ ib}^{a}=\frac{\partial N_{i}^{a}}{%
\partial y^{b}}-\widehat{L}_{\ bi}^{a},\ \widehat{T}_{\ bc}^{a}=\widehat{C}%
_{\ bc}^{a}-\widehat{C}_{\ cb}^{a}.  \label{dtors}
\end{eqnarray}

By straightforward computations, we shall prove that the (nonholonomically
modified) Einstein equations in 5--d gravity can be solved in general form
for the connection $\widehat{\mathbf{D}}.$ Then imposing certain constraints
when $\widehat{\mathbf{D}}\rightarrow \nabla ,$ we shall construct the most
general classes of solutions of gravitational field equations (\ref{einst1})
which can be considered also in general relativity.

\begin{definition}
A distinguished connection $\mathbf{D}$ (in brief, d--connection) on $\ ^{5}%
\mathbf{V}$ is a linear connection preserving under parallelism a
conventional horizontal and vertical splitting (in brief, h-- and
v--splitting) induced by a nonholonomic distribution $\mathbf{N}%
=\{N_{i}^{a}\}$ on tangent bundle
\begin{equation}
T\ ^{5}\mathbf{V=h}\ ^{5}\mathbf{V\oplus }\ v\ ^{5}\mathbf{V.}  \label{wh}
\end{equation}
\end{definition}

We emphasize that the Levi--Civita connection $\nabla ,$ for which $\nabla
\mathbf{g=0}$ and $\mathcal{T}^{\alpha }\doteqdot \nabla \mathbf{e}^{\alpha
}=0,$ is not a d--connection because, in general, it is not adapted to a
N--splitting defined by a Whitney sum (\ref{wh}).

\begin{theorem}
\label{th1}There is a unique canonical d--connection $\widehat{\mathbf{D}}$
satisfying the condition $\widehat{\mathbf{D}}\mathbf{g=}0$ and with
vanishing ''pure'' horizontal and vertical torsion coefficients, i. e. $%
\widehat{T}_{\ jk}^{i}=0$ and $\widehat{T}_{\ bc}^{a}=0,$ see formulas (\ref%
{dtors}).
\end{theorem}

\begin{proof}
It follows by a straightforward verification that
\begin{equation}
\widehat{D}_{j}g_{kl}=0,\widehat{D}_{a}g_{kl}=0,\widehat{D}_{j}h_{ab}=0,%
\widehat{D}_{a}h_{bc}=0,  \label{metcomp}
\end{equation}
i.e. $\widehat{\mathbf{D}}\mathbf{g=0,}$ and computing N--adapted
coefficients of torsion (\ref{dtors}), by using the N--adapted coefficients
\begin{eqnarray}
\widehat{L}_{jk}^{i} &=&\frac{1}{2}g^{ir}\left( \mathbf{e}_{k}g_{jr}+\mathbf{%
e}_{j}g_{kr}-\mathbf{e}_{r}g_{jk}\right) ,  \label{candcon} \\
\widehat{L}_{bk}^{a} &=&e_{b}(N_{k}^{a})+\frac{1}{2}h^{ac}\left( \mathbf{e}%
_{k}h_{bc}-h_{dc}\ e_{b}N_{k}^{d}-h_{db}\ e_{c}N_{k}^{d}\right) ,  \nonumber \\
\widehat{C}_{jc}^{i} &=&\frac{1}{2}g^{ik}e_{c}g_{jk},\ \widehat{C}_{bc}^{a}=%
\frac{1}{2}h^{ad}\left( e_{c}h_{bd}+e_{c}h_{cd}-e_{d}h_{bc}\right) .  \nonumber \end{eqnarray}
(End Proof.)
\end{proof}

\vskip5pt

In general, $\widehat{T}_{\ ja}^{i},\widehat{T}_{\ ji}^{a}$ and $\widehat{T}%
_{\ bi}^{a}$ are not zero, but such nontrivial components of torsion are
induced by some coefficients, depending on off--diagonal terms with $%
N_{i}^{a},$ of a general off--diagonal metric $\mathbf{g}_{\alpha \beta }$ $%
\ $written with respect to a local coordinate basis. Such a torsion $%
\widehat{\mathbf{T}}_{\ \beta \gamma }^{\alpha }$ is very different from
that, for instance, in Einstein--Cartan, string, or gauge gravity when
certain additional field equations (algebraic or dynamical ones) are
considered, see discussions in \cite{vncg,vsgg}. In our case, the nontrivial
torsion coefficients are related to anholonomy coefficients $w_{\alpha \beta
}^{\gamma }$ in (\ref{anhrel}) and none modifications of the usual Einstein
equations in general relativity will be considered.

From Theorem \ref{th1}, one follows:

\begin{corollary}
Any geometric construction for the canonical d--connection $\widehat{\mathbf{%
D}}=\{\widehat{\mathbf{\Gamma }}_{\ \alpha \beta }^{\gamma }\}$ can be
re--defined equivalently into a similar one with the Levi--Civita connection
$\nabla =\{\Gamma _{\ \alpha \beta }^{\gamma }\}$ following formulas
\begin{equation}
\ \Gamma _{\ \alpha \beta }^{\gamma }=\widehat{\mathbf{\Gamma }}_{\ \alpha
\beta }^{\gamma }+\ Z_{\ \alpha \beta }^{\gamma },  \label{deflc}
\end{equation}%
where N--adapted coefficients $\ $of connections, $\Gamma _{\ \alpha \beta
}^{\gamma }$ and $\ \widehat{\mathbf{\Gamma }}_{\ \alpha \beta }^{\gamma },$
and the distortion tensor $\ Z_{\ \alpha \beta }^{\gamma }$ are determined
in unique forms by the coefficients of a metric $\mathbf{g}_{\alpha \beta }.$
\end{corollary}

\begin{proof}
It is similar to that presented for vector bundles in Refs. \cite{ma1987,ma}
but in our case adapted for (pseudo) Riemannian nonholonomic manifolds, see
details in \cite{ijgmmp,vsgg,vrflg}. Here we write down the N--adapted
components of the distortion tensor $Z_{\ \alpha \beta }^{\gamma }$ computed
as
\begin{eqnarray}
\ Z_{jk}^{a} &=&-\widehat{C}_{jb}^{i}g_{ik}h^{ab}-\frac{1}{2}\Omega
_{jk}^{a},~Z_{bk}^{i}=\frac{1}{2}\Omega _{jk}^{c}h_{cb}g^{ji}-\Xi _{jk}^{ih}~%
\widehat{C}_{hb}^{j},  \nonumber \\
Z_{bk}^{a} &=&~^{+}\Xi _{cd}^{ab}~\widehat{T}_{kb}^{c},\ Z_{kb}^{i}=\frac{1}{%
2}\Omega _{jk}^{a}h_{cb}g^{ji}+\Xi _{jk}^{ih}~\widehat{C}_{hb}^{j},\
Z_{jk}^{i}=0,  \label{deft} \\
\ Z_{jb}^{a} &=&-~^{-}\Xi _{cb}^{ad}~\widehat{T}_{jd}^{c},\ Z_{bc}^{a}=0,\
Z_{ab}^{i}=-\frac{g^{ij}}{2}\left[ \widehat{T}_{ja}^{c}h_{cb}+\widehat{T}%
_{jb}^{c}h_{ca}\right] ,  \nonumber \end{eqnarray}%
for $\ \Xi _{jk}^{ih}=\frac{1}{2}(\delta _{j}^{i}\delta
_{k}^{h}-g_{jk}g^{ih}),~^{\pm }\Xi _{cd}^{ab}=\frac{1}{2}(\delta
_{c}^{a}\delta _{d}^{b}+h_{cd}h^{ab})$ and$~\widehat{T}_{\ ja}^{c}=\widehat{L%
}_{aj}^{c}-e_{a}(N_{j}^{c}).\ 
$
\end{proof}

\vskip5pt

In 4--d, the Einstein gravity can be equivalently formulated in the
so--called almost K\" ahler and Lagrange--Finsler variables, as we
considered in Refs. \cite{vrflg,vbrane,vkvnsst,vdqak}. Such types of linear
connections, like $\widehat{\mathbf{\Gamma }}_{\ \alpha \beta }^{\gamma }$
and its nonholonomic deformations, are convenient not only for elaborating
various models of brane and deformation quantization of gravity and
nonsymmetric generalizations but also in constructing general \ solutions of
the Einstein equations for the Levi--Civita connection $\nabla . $

\subsubsection*{Step 3: Nonholonomic deformations of Einstein equations}

In this and next steps, we shall work with the canonical d--connection. We
can compute the nontrivial N--adapted components of curvature of $\widehat{%
\mathbf{D}},$ following formulas
\begin{equation}
\widehat{\mathcal{R}}_{~\beta }^{\alpha }\doteqdot \widehat{\mathbf{D}}%
\widehat{\mathbf{\Gamma }}_{\ \beta }^{\alpha }=d\widehat{\mathbf{\Gamma }}%
_{\ \beta }^{\alpha }-\widehat{\mathbf{\Gamma }}_{\ \beta }^{\gamma }\wedge
\widehat{\mathbf{\Gamma }}_{\ \gamma }^{\alpha }=\widehat{\mathbf{R}}_{\
\beta \gamma \delta }^{\alpha }\mathbf{e}^{\gamma }\wedge \mathbf{e}^{\delta
},  \label{curv}
\end{equation}%
The explicit formulas for the so--called N--adapted coefficients of
curvature $\widehat{\mathbf{\mathbf{R}}}\mathbf{_{\ \beta \gamma \delta
}^{\alpha }=}\{\widehat{R}_{\ hjk}^{i},\widehat{R}_{\ bjk}^{a},\widehat{R}%
_{\ jka}^{i},\widehat{R}_{\ bka}^{c},\widehat{R}_{\ jbc}^{i},\widehat{R}_{\
bcd}^{a}\}$\textbf{\ }of (pseudo) Riemannian spaces are provided, for
instance, in Refs. \cite{ijgmmp,vsgg,vrflg}.

Contracting respectively the N--adapted coefficients of $\widehat{\mathbf{R}}%
_{\ \beta \gamma \delta }^{\alpha }$ (\ref{curv}), one proves that the Ricci
tensor $\widehat{\mathbf{R}}_{\alpha \beta }\doteqdot \widehat{\mathbf{R}}%
_{\ \alpha \beta \tau }^{\tau }$ is characterized by h- v--components, i.e.
the Ricci tensor $\widehat{\mathbf{R}}_{\alpha \beta }=\{\widehat{R}_{ij},%
\widehat{R}_{ia},\ \widehat{R}_{ai},\ \widehat{R}_{ab}\},$
\begin{equation}
\widehat{R}_{ij}\doteqdot \widehat{R}_{\ ijk}^{k},\ \ \widehat{R}%
_{ia}\doteqdot -\widehat{R}_{\ ika}^{k},\ \widehat{R}_{ai}\doteqdot \widehat{%
R}_{\ aib}^{b},\ \widehat{R}_{ab}\doteqdot \widehat{R}_{\ abc}^{c}.
\label{dricci}
\end{equation}%
The scalar curvature of $\widehat{\mathbf{D}}$ is defined
\begin{equation}
\ ^{s}\widehat{R}\doteqdot \mathbf{g}^{\alpha \beta }\widehat{\mathbf{R}}%
_{\alpha \beta }=g^{ij}\widehat{R}_{ij}+h^{ab}\widehat{R}_{ab}.
\label{sdccurv}
\end{equation}

The Einstein tensor of $\widehat{\mathbf{D}}$ is (by definition)
\begin{equation}
\widehat{\mathbf{E}}_{\alpha \beta }=\widehat{\mathbf{R}}_{\alpha \beta }-%
\frac{1}{2}\mathbf{g}_{\alpha \beta }\ ^{s}\widehat{R}.  \label{enstdt}
\end{equation}%
Here, one should be emphasized that tensors $\widehat{\mathbf{\mathbf{R}}}%
\mathbf{_{\ \beta \gamma \delta }^{\alpha },}$ $\widehat{\mathbf{R}}_{\alpha
\beta }$ and $\widehat{\mathbf{E}}_{\alpha \beta }$ (being constructed for
the connection $\widehat{\mathbf{D}}\neq \nabla $) defer by corresponding
distortion tensors from similar tensors $R_{\ \beta \gamma \delta }^{\alpha }%
\mathbf{,}$ $R_{\alpha \beta }$ and $E_{\alpha \beta },$ derived for $\nabla
,$ even both classes of such tensors are completely defined by a same metric
structure $\mathbf{g}_{\alpha \beta }.$ So, the nonholonomically modified
gravitational field equations
\begin{equation}
\widehat{\mathbf{E}}_{\alpha \beta }=\varkappa T_{\beta \delta }
\label{deinst}
\end{equation}%
are not equivalent, in general, to usual Einstein equations for the
Levi--Civita connection $\nabla $ (\ref{einsteq}).\footnote{%
In our previous works \cite{vrflg,vbrane,vkvnsst,vdqak}, we noted that an
equivalence of both types of filed equations would be possible, for
instance, if we introduce a generalized source $\widehat{\mathbf{T}}_{\beta
\delta }$ containing contributions of the distortion tensor (\ref{deft}).}

Nevertheless, it is convenient to use a variant of equations (\ref{deinst}),
\begin{equation}
\widehat{\mathbf{R}}_{\ \ \beta }^{\alpha }=\Upsilon _{\ \ \beta }^{\alpha },
\label{deinst1}
\end{equation}%
with a general source parametrize in the form (\ref{source}), $\Upsilon _{\
\ \beta }^{\alpha }=diag[\Upsilon _{\gamma }],$ because such equations can
be integrated in general form and, for instance, play an important role in
Finsler--Lagrange theories of gravity derived in low energy limits of
string/brane gravity and noncommutative generalizations \cite{vncg}. In Step
5, see below, we shall impose additional constraints on coefficients of
solutions for $\widehat{\mathbf{D}} $ when $\widehat{\mathbf{\Gamma }}_{\
\alpha \beta }^{\gamma }$ will be the same as $\Gamma _{\ \alpha \beta
}^{\gamma },$ with respect to a chosen N--adapted frame (even, in general, $%
\widehat{\mathbf{D}}\neq \nabla ).$\footnote{%
This is possible because the laws of transforms for d--connections, for the
Levi--Civita connection and different types of tensors being adapted, or
not, to a N--splitting (\ref{wh}) are very different.} As a result, we shall
select classes of solutions for equations (\ref{einst1}) with the Ricci
tensor $R_{\alpha \beta }.$

\begin{theorem}
\label{th2}The system of gravitational field equations (\ref{deinst1})
constructed for $\widehat{\mathbf{\Gamma }}_{\ \alpha \beta }^{\gamma }$
with coefficients (\ref{candcon}) and computed for a metric (\ref{gendm})
with coefficients (\ref{data1}), when $g_{\alpha \beta }=diag[\epsilon
_{1},g_{\widehat{i}}(x^{\widehat{k}}),h_{a}(x^{i},v)]$ and $%
N_{k}^{4}=w_{k}(x^{i},v),N_{k}^{5}=n_{k}(x^{i},v),$ is equivalent to this
system of partial differential equations:
\begin{eqnarray}
\widehat{R}_{2}^{2} &=&\widehat{R}_{3}^{3}
 =\frac{1}{2g_{2}g_{3}}[\frac{g_{2}^{\bullet }g_{3}^{\bullet }}{2g_{2}}+%
\frac{(g_{3}^{\bullet })^{2}}{2g_{3}}-g_{3}^{\bullet \bullet }+\frac{%
g_{2}^{^{\prime }}g_{3}^{^{\prime }}}{2g_{3}}+\frac{(g_{2}^{^{\prime }})^{2}%
}{2g_{2}}-g_{2}^{^{\prime \prime }}]=-\Upsilon _{4}(x^{\widehat{i}}),
\label{4ep1a} \\
\widehat{R}_{4}^{4} &=&\widehat{R}_{5}^{5}=\frac{h_{5}^{\ast }}{2h_{4}h_{5}}%
\left( \ln \left| \frac{\sqrt{|h_{4}h_{5}|}}{h_{5}^{\ast }}\right| \right)
^{\ast }=\ -\Upsilon _{2}(x^{i},v),  \label{4ep2a} \\
\widehat{R}_{4i} &=&-w_{i}\frac{\beta }{2h_{4}}-\frac{\alpha _{i}}{2h_{4}}=0,
\label{4ep3a} \\
\widehat{R}_{5i} &=&-\frac{h_{5}}{2h_{4}}\left[ n_{i}^{\ast \ast }+\gamma
n_{i}^{\ast }\right] =0,  \label{4ep4a}
\end{eqnarray}%
where, for $h_{4,5}^{\ast }\neq 0,$\footnote{%
solutions with $h_{4}^{\ast }=0,$ or $h_{5}^{\ast }=0,$ should be analyzed
as some special cases (for simplicity, we omit such considerations in this
work)}%
\begin{equation}
~\phi =\ln |\frac{h_{5}^{\ast }}{\sqrt{|h_{4}h_{5}|}}|,\ \alpha
_{i}=h_{5}^{\ast }\partial _{i}\phi ,\ \beta =h_{4}^{\ast }\ \phi ^{\ast },\
\gamma =\left( \ln |h_{5}|^{3/2}/|h_{4}|\right) ^{\ast }.  \label{auxphi}
\end{equation}
\end{theorem}

\begin{proof}
It can be obtained by straightforward computations as in Parts I and II of
monograph \cite{vsgg}, see also some important details and discussions in Refs. \cite{ijgmmp,vncg}.
\end{proof}

\subsubsection*{Step 4: Solutions with Killing symmetry for nonholonomic
gravitational fields}

The system of equations used in Theorem \ref{th2} can be integrated in very
general forms  for any given $\Upsilon _{2}$ and $%
\Upsilon _{4}.$ Here we note that the equation (\ref{4ep1a})
relates an un--known function $%
g_{2}(x^{2},x^{3})$ to a prescribed $g_{3}(x^{2},x^{3}),$ or inversely. The
equation (\ref{4ep2a}) contains only derivatives on $y^{4}=v$ and allows us
to define $h_{4}(x^{i},v)$ for a given $h_{5}(x^{i},v),$ or inversely, for $%
h_{4,5}^{\ast }\neq 0;$ having defined $h_{4}$ and $h_{5},$ we can compute
the coefficients (\ref{auxphi}), which allows us to find $w_{i}$ from
algebraic equations (\ref{4ep3a}) and to compute $n_{i}$ by integrating
two times on $v$ as follow from equations (\ref{4ep4a}). This way, we prove:

\begin{proposition}
\label{pr1}The general class of solutions of nonholonomic gravitational
equations (\ref{deinst1}) with one Killing symmetry on $e_{5}=\partial
/\partial y^{5}$ is defined by an ansatz (\ref{ansgensol}) with $\omega
^{2}=1$ and coefficients $g_{\widehat{i}},h_{a},$ $w_{k},n_{k}$ computed
following formulas (\ref{coeff}).
\end{proposition}

We note that such classes of solutions are very general ones and contain as
particular cases all possible exact solutions for (non) holonomic Einstein
spaces with Killing symmetry. They also can be generalized to include
arbitrary finite sets of parameters as we considered in Ref. \cite{ijgmmp}.

\subsubsection*{Step 5: Constraints generating solutions in Einstein gravity}

Nevertheless, the solutions constructed following Proposition \ref{pr1} are
for the canonical d--connection, $\widehat{\mathbf{D}},$ and not for the
Levi--Civita one, $\nabla .$ We can see that both the torsion $\widehat{%
\mathbf{T}}_{\ \beta \gamma }^{\alpha }$ (\ref{dtors}) and distortion tensor
$Z_{\ \alpha \beta }^{\gamma }$ (\ref{deft}) became zero if and only if
\begin{equation}
\widehat{C}_{jb}^{i}=0,\Omega _{\ ji}^{a}=0,\widehat{T}_{ja}^{c}=0,
\label{lccondg}
\end{equation}%
with respect to a N--adapted basis (in general, such a basis is anholonomic
because $w_{ia}^{b}=\partial _{a}N_{i}^{b}$ is not obligatory zero, see
formulas (\ref{anhrel})). In such a case, the distortion relations (\ref%
{deflc}) transform into $\ \Gamma _{\ \alpha \beta }^{\gamma }=\widehat{%
\mathbf{\Gamma }}_{\ \alpha \beta }^{\gamma }.$

\begin{corollary}
\label{cor1} An ansatz (\ref{ansgensol}) with $\omega ^{2}=1$ and
coefficients $g_{\widehat{i}},h_{a},$ $w_{k},n_{k}$ computed following
formulas (\ref{coeff}) defines solutions with one Killing symmetry on $%
e_{5}=\partial /\partial y^{5}$ of the Einstein equations (\ref{einst1}) for
the Levi--Civita connection $\Gamma _{\ \alpha \beta }^{\gamma },$ all
formulas being considered with respect to N--adapted frames, if the
coefficients of metric are subjected additionally to the conditions (\ref%
{lccond}).
\end{corollary}

\begin{proof}
By straightforward computations for ansatz defined by metrics (\ref{gendm})
with coefficients (\ref{data1}), we get that the conditions (\ref{lccondg})
resulting in $\Gamma _{\ \alpha \beta }^{\gamma }=\widehat{\mathbf{\Gamma }}%
_{\ \alpha \beta }^{\gamma }$ are just those written as (\ref{lccond}).\
\end{proof}

\vskip2pt

Steps 1--5 considered above result in:
\vskip1pt
{\bf Conclusion:\ }
In order to generate exact solutions with Killing symmetry in Einstein
gravity and its 5--d extensions, we should consider N--adapted frames and
nonholonomic deformations of the Levi--Civita connection to an auxiliary
metric compatible d--connection (for instance, to the canonical
d--connection, $\widehat{\mathbf{D}}),$ when the corresponding system of
nonholonomic gravitational field equations (\ref{4ep1a})--(\ref{4ep4a}) can
be integrated in general form. Subjecting the integral variety of such
solutions to additional constraints of type (\ref{lccondg}), i.e. imposing
the conditions (\ref{lccond}) to the coefficients of metrics, we may
construct new classes of exact solutions of Einstein equations for the
Levi--Civita connection $\nabla .$

\subsubsection*{Step 6: General solutions in Einstein gravity}

The last step which allows us to consider the most general classes of
solutions of the nonholonomic gravitational field equations (\ref{deinst1}),
and (for more particular cases), of Einstein equations (\ref{einst1}), is to
extend the anholonomic deformation  method to the case of metrics depending on all
coordinates $u^{\alpha }=(x^{i},y^{a}),$ i.e. to solutions without any
prescribed Killing symmetry.

Let us introduce a nontrivial multiple $\omega ^{2}(x^{i},y^{a})$ before
coefficients $h_{a}$ in a metric (\ref{gendm}), when the rest of
coefficients are parametrized in the form (\ref{data1}). We get an ansatz of
type%
\begin{eqnarray}
^{\omega }\mathbf{g} &=&\epsilon _{1}e^{1}\otimes e^{1}+g_{\widehat{j}}(x^{%
\widehat{k}})e^{\widehat{j}}\otimes e^{\widehat{j}}+\omega
^{2}(x^{i},y^{a})h_{a}(x^{i},v)\mathbf{e}^{a}\otimes \mathbf{e}^{a},  \nonumber \\
\mathbf{e}^{4} &=&dy^{4}+w_{i}(x^{k},v)dx^{i},\mathbf{e}%
^{5}=dy^{5}+n_{i}(x^{k},v)dx^{i}.  \label{genansc}
\end{eqnarray}%
Introducing coefficients of $^{\omega }\mathbf{g}$ into formulas (\ref%
{candcon}), we compute $\ ^{\omega }\widehat{\mathbf{\Gamma }}_{\ \alpha
\beta }^{\gamma },$ which allows us to define, see (\ref{dricci}), $\
^{\omega }\widehat{\mathbf{R}}_{\alpha \beta }=\{\widehat{R}_{ij},\widehat{R}%
_{ia},\ \ ^{\omega }\widehat{R}_{ai},\ \ ^{\omega }\widehat{R}_{ab}\}.$

\begin{lemma}
\label{lm1}For a generalized ansatz (\ref{genansc}), which for $\omega
^{2}=1 $ is a solution of \ nonholonomic gravitational equations (\ref%
{deinst1}) with Killing symmetry on $e_{5}=\partial /\partial y,$ we obtain
\begin{equation}
\ ^{\omega }\widehat{R}_{\ b}^{a}=\ \widehat{R}_{\ b}^{a}+\ ^{\omega }%
\widehat{Z}_{\ b}^{a}\mbox{\ and }\ \ ^{\omega }\widehat{R}_{ai}=\ \
\widehat{R}_{ai}=0,  \label{eqconf}
\end{equation}%
with $\ ^{\omega }\widehat{Z}_{\ b}^{a}=diag[\ ^{\omega }\widehat{Z}%
_{c}(x^{i},y^{a})]$ $\ $determined for any $\omega ^{2}(x^{i},y^{a})$
subjected to conditions $\mathbf{e}_{k}\omega =0$ (\ref{confcond}) and $%
\widehat{T}_{ja}^{c}=0$.
\end{lemma}

\begin{proof}
By straightforward computations for ansatz defined by metrics (\ref{genansc}%
), we see that the v--part containing coefficients $\omega ^{2}h_{a}$
results in certain two dimensional conformal transforms of $\ \widehat{R}%
_{ab}$ to $\ ^{\omega }\widehat{R}_{ab}$ (we can consider in this case any
fixed values $x^{i}$ but arbitrary coordinates $y^{4}$ and $y^{5})$ and
certain additional terms to $\ \widehat{R}_{ai}$ giving a nonzero $\
^{\omega }\widehat{R}_{ai}.$ Nevertheless, we can satisfy the equations (\ref%
{eqconf}) with $\widehat{T}_{ja}^{c}=0$ for any nontrivial factor $\omega $
for which $\mathbf{e}_{k}\omega =0.$ Of course, for such nontrivial $\
^{\omega }\widehat{Z}_{c}(x^{i},y^{a}),$ we should redefine the sources (\ref%
{source}), via frame/coordinate transform, which would allow us to solve
equations of type (\ref{4ep2a}), when $\widehat{R}_{4}^{4}=$ $\ _{4}^{\omega
}\Upsilon (x^{i},v)\ $\ and $\widehat{R}_{5}^{5}=$ $\ _{5}^{\omega }\Upsilon
(x^{i},v), $ with contributions to the vertical conformal transforms, are
equivalent to certain $\widehat{R}_{4}^{4}=\widehat{R}_{5}^{5}=\ -\Upsilon
_{2}(x^{i},v).$ For constraints of type $\widehat{T}_{ja}^{c}=0$ and $%
\mathbf{e}_{k}\omega =0,$ and dimensional vertical subspaces, one holds $%
\widehat{R}_{4}^{4}=\widehat{R}_{5}^{5}=\ ^{\omega }\widehat{R}_{4}^{4}=\
^{\omega }\widehat{R}_{5}^{5}.$
\end{proof}

Summarizing the results of Theorems \ref{th1} and \ref{th2}, Proposition \ref%
{pr1}, Corollary \ref{cor1} and Lemma \ref{lm1}, we prove the Main Result
stated in Theorem \ref{mth}.

As a matter of principle, any exact solution in gravity theories (Einstein
gravity and sting/ brane/ gauge/ Kaluza--Klein, Lagrange--Finsler,
supersymmetric and/or noncommutative generalizations etc) can be represented
in a form (\ref{ansgensol}) or certain nonholonomic frame transforms/
deformations with extra--dimension coordinates and various types of
commutative and noncommutative parameters, see more general/ alternative
constructions in Refs. \cite{vncg,vsgg,vsolhd,vnbh}. Perhaps, the anholonomic deformation
method allows us to construct general solutions of gravitational equations
in the form (\ref{einst1}), for arbitrary dimension and source (\ref{source}%
), when the Ricci tensor is determined by any generalized linear and
nonlinear connections\footnote{%
of course, the term ''general solution'' should be used in a quite
approximate form because it may be not clear how to define a ''general
unique solution'' in a rigorous mathematical form for some nonlinear systems
of equations with possible singularities of coefficients and/or generalized
topological and group symmetries etc}. The length of this paper does not
allow us to speculate on symmetries and properties of such solutions and
possible physical implications (for instance, how to consider black hole and
cosmological solutions with singularities and horizons, and their
nonholonomic deformations); for details and discussions, we send the reader
to Refs. \cite{ijgmmp,vncg,vsgg,vrflg}.

\vskip5pt

\textbf{Acknowledgement: } Author thanks M. Anastasiei for important discussions and support.

\vskip5pt

\textbf{Remarks on submissions and publications} This preprint version is almost identic to a published  letter variant (see: S. Vacaru, IJGMMP \textbf{8} (2011) 9-21; submitted to arXiv.org on September 22, 2009). It is also related to another already published article \cite{vsolhd} (further variants 2-4, extending the version v1, were put as   arXiv: 0909.3949 [gr-qc] beginning October 1, 2009) containing detailed proofs and generalizations of results on exact solutions of Einstein equations for (pseudo) Riemannian spaces of arbitrary finite dimension $n+m > 5$. Following a discussion and suggestion of arXiv's Moderator (from October 2009), we submitted two variants of electronic preprints because two different manuscripts were published in different journals (with different titles and lengths and rather different contents) and the letter variant may have certain priorities for readers interested in exact solutions in general relativity but  not in extra dimension generalizations. On June 20, 2011, moderators of arXiv.org decided  to provide a different number to the "short" variant of paper as a submission to physics.gen-ph.

\end{document}